\pdfoutput=1
\newif\ifPublic
\Publictrue
\PassOptionsToPackage{pdfa}{hyperref}
\ifPublic
\documentclass[sigplan,screen,pbalance=true,nonacm=true,authorversion=true]{acmart}
\else
\documentclass[sigplan,screen,anonymous]{acmart}
\fi
\hypersetup{pdfapart=1, pdfaconformance=b}
\AtBeginDocument{%
  }

\usepackage{xspace}
\usepackage{tikz}
\usepackage{mathpartir}
\usepackage{cancel}
\usepackage{mathtools}
\usepackage{listings}
\newcommand{\power}[1]{\ensuremath{\mathbb{P}(#1)}}

\newcommand{\bind}{\mathbin{>\!\!>\mkern-5.9mu=}}
\newcommand{\seq}{\mathbin{>\!\!>}}

\setcopyright{cc}
\setcctype{by-nc-nd}
\acmDOI{10.1145/3840563.3843706}
\acmYear{2026}
\copyrightyear{2026}
\acmISBN{979-8-4007-2933-1/2026/10}
\acmConference[NSAD '26]{Proceedings of the 11th ACM SIGPLAN International Workshop on Numerical and Symbolic Abstract Domains}{October 4--9, 2026}{Oakland, CA, USA}
\acmBooktitle{Proceedings of the 11th ACM SIGPLAN International Workshop on Numerical and Symbolic Abstract Domains (NSAD '26), October 4--9, 2026, Oakland, CA, USA}
\acmSubmissionID{splashws26nsadmain-p6-p}
\received{2026-06-10}
\received[accepted]{2026-08-04}

\begin{document}

\title{Effect Systems as Abstract Interpretations}

\author{Colin S. Gordon}
\correspondingauthor
\orcid{0000-0002-9012-4490}
\affiliation{%
  \institution{Drexel University}
  \city{Philadelphia}
  \country{USA}
}
\email{csgordon@drexel.edu}

\renewcommand{\shortauthors}{Gordon}

\begin{abstract}
Many forms of static reasoning about program behaviours are known in the literature, yet
formal relationships are studied surprisingly infrequently.
While most type systems are well-known to be captured by abstract interpretations,
the situation for type-and-effect systems is, in the general case, unsettled despite
strong hypotheses and occasional framing of effect systems as abstract interpretations.

We develop a formal relationship between abstract interpretations and a general class
of effect systems.
First, we describe an embedding of effect quantales into abstract domains.
Second, we recover the general form of an effect system as an abstract interpretation ---
not on states or values, but on event occurrences.
\end{abstract}

\maketitle

\section{Introduction}
\label{sec:intro}
Static type systems have long been viewed as effectively a kind of abstract interpretation~\cite{cousot1997types}.
Implicit in this view is the belief that this applies to effect systems as well. A number of effect systems have been proposed which effectively perform abstract interpretations~\cite{Skalka2008,skalka2008types,holdermans2010polyvariant,ivaskovic2020dataflow}, sometimes explicitly formalizing effects as abstract domains~\cite{hofmann2014abstract,nicola2025abstract,skalka2020types,tang1994separate}.
\citet{vouillon1995type} even proved equivalence between an abstract interpretation and an effect system for a time complexity analysis to explore the two techniques' similarities.
But this relationship has not been systematically studied before in a way that establishes general relationships, rather than relationships (either equivalence or defining effects via abstract interpretation) for specific systems.
We partially systematize the techniques' relationship, describing a general construction of an abstract domain of trace properties from an arbitrary \emph{effect quantale}~\cite{gordon2021polymorphic}, thus showing that \emph{many} effect systems can be reconstructed as abstract interpretations.
\looseness=-1

This draws a broader formal connection between two approaches to static analysis.
Effect quantales as a framework have, like abstract interpretation's domains, been used both to guide the design of new analyses~\cite{saffrich2022relating,bao2021reachability},
and as a basis for transforming an effect system into one with improved precision~\cite{gordon2023error} or additional constructs such as loops~\cite{gordon2021polymorphic} or continuations~\cite{gordon2020lifting}.
The relationship also raises questions about the relationship between abstract interpretation and the \emph{logical approach to type soundness}~\cite{timany2024logical}, where typing judgments denote judgements in a program logic (typically separation logic).
The earliest systematic version of this approach~\cite{dinsdale2013views}, used to prove type safety for a data-race-free C\# dialect~\cite{gordon2012uniqueness},
involved giving denotations of types as sets with additional properties, in a way highly reminiscent of Cousot's~\cite{cousot1997types} concretization of types as program properties.
Better understanding the relationship between effect systems and abstract interpretation may permit transfer of constructions in both directions, and may open new ways to study soundness or inference for effect systems.

\section{Preliminaries}
\label{sec:prelim}
We assume familiarity with abstract interpretation based on Galois connections~\cite{cousot1979systematic,cousot1992abstract},
and focus on recalling effect systems in general, and effect quantales in particular.

Effect systems have a lengthy history~\cite{gifford86,lucassen88} as a way to augment type systems to reason about side effects of interest in a program: which regions of memory are accessed~\cite{lucassen88,bocchino09}, which locks' data are accessed~\cite{objtyrace99,rccjava00}, which synchronization actions are performed~\cite{gordon2012static,suenaga2008type,flanagan2003atomicity}, and of course the classic checked exceptions~\cite{benton2007exceptions}.
Most effect systems ignore program order, tracking only sets of things that could occur during execution. But even fairly early, the notion of an effect system was extended to deal with program order, initially tracking the order of messages sent in Concurrent ML~\cite{nielson1993cml,nielson1994constraints,amtoft1999}, and later other concurrency properties~\cite{flanagan2003atomicity,boyapati02,gordon2012static,suenaga2008type,ivaskovic2020graded},
general temporal properties~\cite{Koskinen14LTR,sekiyama2023temporal,hofmann2014abstract}, dataflow analysis~\cite{ivaskovic2020dataflow,holdermans2010polyvariant}, and more.

Effect systems that ignore program order typically have effects that form a bounded join semilattices~\cite{marino09} --- behaviours of any two program fragments are combined by over-approximating both with a join.
For order-aware effect systems, sometimes called flow-sensitive or \emph{sequential}~\cite{tate2013sequential}, a variety of approaches exist; see Gordon~\cite{gordon2021polymorphic} for a survey.
\emph{Effect quantales}~\cite{gordon2017generic,gordon2021polymorphic} are the most fully-developed framework, with
systematic approaches to deal with iteration~\cite{gordon2021polymorphic}, arbitrary control flow~\cite{gordon2020lifting}, and precise error reporting~\cite{gordon2023error}.
\looseness=-1

\begin{definition}[Effect Quantale]
An effect quantale is a carrier set $Q$ equipped with a partial binary join $\sqcup$ and a partial monoid $\rhd$ with unit $I$, which additionally satisfies the distributive laws
$x\rhd(y\sqcup z)=(x\rhd y)\sqcup(x\rhd z)$
and
$(x\sqcup y)\rhd z=(x\rhd z)\sqcup(y\rhd z)$
\end{definition}

We identify the algebraic structure with its carrier set when no ambiguity results.
$\sqcup$ is used to over-approximate unresolved finite branching (e.g., conditionals), while $\rhd$ is used to model program ordering (e.g., sequencing operations according to evaluation order). Note that a bounded join semilattice is a (total) effect quantale when reusing the join as the sequencing operation.
The partial operations model early rejection of locally-invalid operations, such as repeated acquisition of a non-reentrant lock.
In every effect quantale, we obtain a partial order $x\sqsubseteq y \Leftrightarrow x\sqcup y = y$.

\subsection{A Simple Language with Effects}
Effect quantales allow extending type systems with effect tracking, using the different operators to capture or discard program evaluation ordering.
Figure \ref{fig:lang} shows a small call-by-value lambda calculus with non-deterministic branching.
The judgment $\Gamma \vdash e : \tau \mid \chi$ can be read as roughly ``under local variable typing assumptions $\Gamma$, terminating executions of $e$ will produce a value of type $\tau$ while causing side effects summarized by $\chi$.''
Notice that function types carry the \emph{latent effect} of the function, tracking the function body's behaviours which do not occur until it is invoked --- captured  in the type rule for function application, which sequences (via $\rhd$) the behaviours
of the left and right expressions, followed by the latent effect.
Primitive behaviours are raised, for simplicity, by a construct $\mathsf{ev}(-)$ which raises a specific effect $\chi$.
Recursive functions can encode common looping constructs:
\looseness=-1
\[ \mathsf{while*}\;e \equiv (\mathsf{rec}_f x:\mathsf{unit}\ldotp \mathsf{if*}\;x\;(f\;x) )\;\mathsf{tt} \]

\begin{figure}\scriptsize 
\begin{mathpar}
\inferrule[\scriptsize T-Unit]{ }{\Gamma\vdash \mathsf{tt} : \mathsf{unit} \mid I}
\and
\inferrule[\scriptsize T-Var]{\Gamma(x)=\tau}{\Gamma\vdash x : \tau \mid I}
\and
\inferrule[\scriptsize T-Lambda]{\Gamma,x:\tau\vdash e : \tau' \mid \chi}{\Gamma\vdash (\lambda x:\tau\ldotp e) : \tau\xrightarrow{\chi}\tau' \mid I}
\and
\inferrule*[left=\scriptsize T-App]{
  \Gamma\vdash e_1 : \tau\xrightarrow{\chi} \tau' \mid \chi_1 \\
  \Gamma\vdash e_2 : \tau \mid \chi_2 
}{\Gamma\vdash e_1~e_2 : \tau' \mid \chi_1\rhd\chi_2\rhd\chi}
\and
\inferrule*[left=\scriptsize T-Ev]{ }{
  \Gamma\vdash \mathbf{ev}(\chi) : \mathsf{unit} \mid \chi
}
\and
\inferrule[\scriptsize T-Branch]{
  \Gamma\vdash e_1 : \tau \mid \chi_1\\
  \Gamma\vdash e_2 : \tau \mid \chi_2
}{
  \Gamma\vdash\mathsf{if*}\;e_1\;e_2 : \tau \mid \chi_1\sqcup\chi_2
}
\and
\inferrule[\scriptsize T-Rec]{\Gamma,f:\tau\xrightarrow{\chi}\tau',x:\tau\vdash e : \tau' \mid \chi}{\Gamma\vdash (\mathsf{rec}_f x:\tau\ldotp e) : \tau\xrightarrow{\chi}\tau' \mid I}
\and
\inferrule*[left=\scriptsize T-Sub]{\Gamma\vdash e : \tau \mid \chi \\ \chi \sqsubseteq \chi'}{\Gamma\vdash e : \tau \mid \chi'}
\end{mathpar}
\[
\begin{array}{r@{\;}c@{\;}l@{\;}l}
\mathsf{ev}(\chi) & \xhookrightarrow{\chi} & \mathsf{tt}\\
(\lambda x:\tau\ldotp e)\;v & \xhookrightarrow{I} & e[x\mapsto v]\\
(\mathsf{rec}_f x:\tau\ldotp e)\;v & \xhookrightarrow{I} & e[x\mapsto v,f\mapsto(\mathsf{rec}_f x:\tau\ldotp e)]\\
e_1\;e_2&\xhookrightarrow{\chi}&e_1'\;e_2&\textrm{if}~e_1\xhookrightarrow{\chi}e_1'\\
v_1\;e_2&\xhookrightarrow{\chi}&e_1\;e_2'&\textrm{if}~e_2\xhookrightarrow{\chi}e_2'\\
\mathsf{if*}\;e_1\;e_2&\xhookrightarrow{I}&e_1\\
\mathsf{if*}\;e_1\;e_2&\xhookrightarrow{I}&e_2
\end{array}
\]
\vspace{-1em}
\caption{A simple effect language}
\vspace{-1em}
\label{fig:lang}
\end{figure}

The vast majority of effect systems are insensitive to values, so non-deterministic branching captures the core assumptions of
effect quantales as well as other semantic (categorical) generic models of effect systems.
The use of the event primitive $\mathsf{ev}(-)$ may be more surprising from an abstract interpretation perspective, despite being quite standard in effect systems, because it means the events are
not obviously connected to program semantics.
Many effect systems reason about sequences of external interactions (for which no formal semantics exist) or about internal steps such as throwing an exception (in which case tagging throws is sufficient to distinguish throws inside a corresponding \texttt{try-catch}).
With more space we could pursue parameterize the language by state and update primitives~\cite{gordon2021polymorphic}, with a formal interpretation of effects as relations on states.
\looseness=-1

Type safety for this calculus can be obtained via standard syntactic techniques~\cite{gordon2021polymorphic},
yielding a guarantee that if

\[\cdot\vdash e_0 : \tau \mid \chi
\quad\textrm{and}\quad
e_0 \xrightarrow{\chi_1} e_1 \xrightarrow{\chi_2} e_2 \ldots \xrightarrow{\chi_{n-1}} e_n
\]
then there exists a $\chi_n$ such that $\cdot\vdash e_n : \tau \mid \chi_n$ and
$\chi_1\rhd\chi_2\rhd\ldots\rhd\chi_{n-1}\rhd\chi_n \sqsubseteq \chi$
(which, implies that the sequential composition is defined and therefore valid).
We will develop type safety results using abstract interpretation instead.
\looseness=-1

\begin{figure}[b]\scriptsize
\begin{center}
\begin{tikzpicture}[node distance=1.5cm]
\node(neutral){\fbox{$\varepsilon$}};
\node(critical)[above left of=neutral]{\fbox{$\mathsf{critical}$}};
\node(locking)[left of=critical]{\fbox{$\mathsf{locking}$}\hspace{2em}};
\node(entrant)[above right of=neutral]{\fbox{$\mathsf{entrant}$}};
\node(unlocking)[right of=entrant]{${}\qquad{}$\fbox{$\mathsf{unlocking}$}};
\draw(neutral)--(critical);
\draw(neutral)--(entrant);
\end{tikzpicture}

$
\begin{array}{|c|ccccc|}
\hline
\rhd & \mathsf{locking} & \mathsf{unlocking} & \mathsf{critical} & \mathsf{entrant} & \varepsilon \\
\hline
\mathsf{locking} & - & \mathsf{entrant} & \mathsf{locking} & - & \mathsf{locking} \\
\mathsf{unlocking} & \mathsf{critical} & - & - & \mathsf{unlocking} & \mathsf{unlocking} \\
\mathsf{critical} & - & \mathsf{unlocking} & \mathsf{critical} & - & \mathsf{critical} \\
\mathsf{entrant} & \mathsf{locking} & - & - & \mathsf{entrant} & \mathsf{entrant} \\
\varepsilon & \mathsf{locking} & \mathsf{unlocking} & \mathsf{critical} & \mathsf{entrant} & \varepsilon \\
\hline
\end{array}
$
\end{center}
\caption{Tate's critical section effects as an effect quantale \textsc{Crit}. $-$ represents an undefined result for composition.}
\vspace{-1em}
\label{fig:crit-effects}
\end{figure}
\subsection{Examples of Effect Systems}
Most readers are likely familiar with the most widely-used effect system: \textsc{Java}'s checked exceptions, where effects are the set of exceptions
that may be thrown. Some are likely familiar with the modern twist of an effect system tracking which algebraic operations may be invoked~\cite{bauer2014effect}.
But we are concerned with the generalization to cases where order of operations, not only their enclosing constructs, are important.
Figures \ref{fig:crit-effects} recalls Tate's~\cite{tate2013sequential} effect system for tracking use of non-reentrant critical sections --- those where recursive acquisition is unsupported, and therefore an error. If acquisition \lstinline|acq()| has effect \textsf{locking} and release \lstinline|rel()| has effect \textsf{unlocking}, this effect system classifies code that acquires; releases; acquires-and-releases (\textsf{entrant}); assumes-and-leaves-acquired (\textsc{critical}, for code which releases and reacquires or uses the mutually-exclusive resource); and code which is insensitive to the critical section ($\varepsilon$).
Note that repeated locking or unlocking is undefined, as is
the join of effects corresponding to different locking behaviour --- as with regular type systems, effect systems perform both the analysis and rejection of unsatisfactory analysis results (by declining to assign \emph{any} effect to ill-behaved code).

\section{Abstracting Trace Semantics}
We begin by developing denotational semantics for the language of the previous section,
loosely following the approach of \citet{cousot1997types} in abstracting directly
from denotational semantics.
As noted above, unlike the traces most often considered in abstract interpretation,
effect systems are often concerned with sequences of events, rather than states.

This is most explicit in the concrete semantics, where
computations produce values alongside a trace of events triggered/emitted in the process of producing that value.
Thus we model expression semantics as (meta-)functions from
environments to computations of value-trace pairs.
Because our language treats branching non-deterministically,
computations produce \emph{sets} of these pairs --- different
branches may yield potentially different values and traces.
\[
\begin{array}{r@{\;}c@{\;}l@{\;}l}
\mathbb{W}&=&\{\omega\} & \textrm{error}\\
\mathbb{U} &=& \mathsf{tt} & \textrm{unit}\\
f\in\mathbb{V}&\simeq&\mathbb{W}_\bot\oplus \mathbb{U}_\bot\oplus[\mathbb{V}\rightarrow\power{\mathbb{V}\times Q^*}]_\bot & \textrm{values}\\
R\in\mathbb{R}&=&\mathbb{X}\rightarrow\mathbb{V} & \textrm{environments}\\
\phi\in\mathbb{S}&=&\mathbb{R}\rightarrow\power{\mathbb{V}\times Q^*}&\textrm{sem.\ domains}
\end{array}
\]
Note that values are distinct from computations: the result of  a computation is drawn from \power{\mathbb{V}\times Q^*}, while
values do not carry traces. The portion of the value domain used for the semantics of closures produces are (meta-)functions from values to computation results.
The combination of non-determinism and higher-order functions leads to domain equations with no solutions in $\mathbf{Set}$.
These equations have solutions in various categories of metric spaces~\cite{america1989solving,birkedal2010category}, which also ensure the existence of the least fixed points used in what follows.
Our results do not depend on any specific category, as long as the equations have solutions.
For brevity we continue to refer to elements of the powerdomains $\mathbb{P}(-)$ as sets.
\looseness=-1

\newcommand{\csem}[1]{\ensuremath{\mathbf{S}\llbracket{#1}\rrbracket}\xspace}
We assume our metalanguage has monadic structure for handling undefinedness or errors, and for mapping over collections.
In particular we use
\[
(x,T) \seq (y,T') \equiv ((x = \bot \vee x = \Omega) \mathbin{?} \langle x,T\rangle \mathbin{:} \langle y,T\cdot T'\rangle)
\]
to concisely propagate invalid computations, and
for $S\in\power{A\times Q^*}$ and a function $f:A\rightarrow\power{B\times Q*}$
\[
S \bind f \equiv \bigcup\{ f\;a \mid a \in S \}
\]
We also pattern-match on pairs when writing metalanguage functions (denoted by $\Lambda$ below), rather than using projections. We write $\Omega$ for the injection of $\omega$ into the lifted domain.
\looseness=-1

\noindent\hspace{-0.5em}$
\begin{array}{r@{}c@{}l}
\csem{-}&:& \mathsf{Expression} \rightarrow \mathbb{S}\\
\csem{\mathsf{tt}}&=& \Lambda R\ldotp \{\langle\uparrow\mathsf{tt},I\rangle\}\\
\csem{x}&=&\Lambda R\ldotp \{\langle R(x), I\rangle\}\\
\csem{\lambda x:\tau\ldotp e}&=&
  \Lambda R\ldotp
    \{\langle\uparrow(\Lambda \hat{x}\ldotp \hat{x} \seq (\csem{e}(R[x\leftarrow \hat{x}]))),I\rangle\}\\
\csem{e_1\;e_2}&=&
  \Lambda R\ldotp \csem{e_1}(R) \bind (\Lambda \langle f, T_1\rangle\ldotp\\
  &&\qquad (f=\bot\vee f=\Omega)\mathbin{?} \langle f, T_1\rangle\mathbin{:}\\
  &&\qquad \csem{e_2}(R) \bind (\Lambda\langle a,T_2\rangle\ldotp\\
  &&\qquad\quad (a=\bot\vee a=\Omega)\mathbin{?}\langle a,T_1\cdot T_2\rangle\mathbin{:}\\
  &&\qquad\quad (\mathsf{nonfunc}(f) \mathbin{?}\langle\Omega,T_1\cdot T_2\rangle  \mathbin{:}\\
  &&\qquad\quad\mathsf{let}\;\langle u,T_3\rangle = f\;a\;\mathsf{in}\; \langle u,T_1\cdot T_2\cdot T_3\rangle)))\\
\csem{\mathbf{ev}(\chi)}&=& \Lambda R\ldotp \{\langle\uparrow\mathsf{tt},\chi\rangle\}\\
\csem{\mathsf{if*}\;e_t\;e_f}&=&
  \Lambda R\ldotp \csem{e_t}(R)\cup\csem{e_f}(R)\\
\csem{\mathsf{rec}_f x:\tau\ldotp e}&=&
  \Lambda R\ldotp
    \mathsf{lfp}_{\uparrow(\Lambda\hat{x}\ldotp \emptyset)}
      \left(\Lambda\varphi\ldotp \csem{\lambda x:\tau\ldotp e}(R[f\leftarrow\varphi])\right)
\end{array}
$

Again, while it may seem odd for the semantics to directly mention effects (which are used to approximate semantics), this actually matches typical usage. Often for an effect quantale $Q$ there is a subset $C\subseteq Q$ of \emph{concrete} effects which may arise in executions (e.g., raising a specific single exception) with the rest of the effect quantale ($Q\setminus C$) conceptually grouping them (e.g., tracking sets of thrown exceptions).\footnote{This is not unlike the situation with Kleene Algebras with Tests (\textsc{KAT}s)~\cite{kozen1997kleene}, where the tests are a particular subset of the carrier set.
These concrete effects often play a role similar to the order-theoretic notion of atoms (minimal non-bottom elements) in abstract interpretation.}

\paragraph{Collecting Semantics}
Our concrete semantics is already non-deterministic, but still assigns
only a single semantic function to an expression. Thus we still define a separate collecting semantics in \power{\mathbb{S}}:
\newcommand{\col}[1]{\ensuremath{\mathbf{C}\llbracket{#1}\rrbracket}\xspace}
\[
\begin{array}{r@{\;}c@{\;}l}
\col{-}&:&\mathbb{E}\rightarrow\power{\mathbb{S}}\\
\col{e}&=&\{\csem{e}\}
\end{array}
\]

\section{Concretizing Effect Quantales}
As effects are intended to summarize possible program behaviours,
we can (begin to) treat them as abstract domains.
Eventually we must treat types and typing environments as well, and relate all of these elements to semantics. But for now we focus on effects as describing sets of traces: %
\[
\begin{array}{r@{\;}c@{\;}l}
\gamma_0(-)&:&Q\rightarrow (\power{Q^*})\\
\gamma_0(\chi)&=&\{T \mid \mathsf{fold}\;{\rhd}\;I\;T\sqsubseteq\chi\}
\end{array}
\]
Thus every effect $\chi$ abstracts any and all event traces whose sequential composition (under $\rhd$) is bounded by ($\sqsubseteq$) $\chi$.
Abstraction becomes more challenging, because there are traces for which the \textsf{fold} operation is undefined.

Tate's effect system~\cite{tate2013sequential} (Fig.\ \ref{fig:crit-effects}) for non-reentrant critical section tracking
is a compact example of why effect quantale operations are partial: certain sequences (such as acquiring a non-reentrant lock twice in a row, represented as $\mathsf{locking}\rhd\mathsf{locking}$) are considered errors, and by not allowing such compositions to yield an effect, there is no effect to assign to code containing the corresponding error.

However our concrete semantics permits such misorderings:
for any semantic environment $R$,
using $e_1;_\tau e_2$ as shorthand for $(\lambda x:\tau\ldotp e_2)\;e_1$,
\[\csem{\mathsf{ev}(\mathsf{locking});_\mathsf{unit} \mathsf{ev}(\mathsf{locking})}(R) =\{\langle\mathsf{tt},\mathsf{locking}\cdot\mathsf{locking}\rangle\}\]
This computation is --- by design --- \emph{not} abstracted by any effect in \textsc{Crit}.
A similar issue arises with undefined joins:
\[\begin{array}{l}
\csem{\mathsf{if*}\;(\mathsf{ev}(\mathsf{locking}))\;(\mathsf{ev}(\mathsf{unlocking}))}(R) \\
\qquad=\{\langle\mathsf{tt},\mathsf{locking}\rangle,\langle\mathsf{tt},\mathsf{unlocking}\rangle\}
\end{array}\]

There are two plausible approaches to handling this:
partial Galois connections~\cite{kappelmann2023transport,dagand2018foundations},
and completions~\cite{macneille1937partially,schmidt2012inverse}.
While partial Galois connections are not new, completing the partial order structure allows applying
the rich existing literature on abstract interpretations without further modification.
Thus we complete an effect quantale to a complete lattice suitable for abstract interpretation in two steps:
adding a completion of the monoid structure to classify undesirable traces,
and completing the order. The order completion is standard but slightly more involved, so we begin there.

\subsection{(Dedekind-)MacNeille Completions}

\newcommand{\upperbounds}[1]{\ensuremath{{(#1)^{u}}}}
\newcommand{\lowerbounds}[1]{\ensuremath{{(#1)^{\ell}}}}

For a partial order $Q$ (e.g., an effect quantale), and $S\subseteq Q$,
let
\[
\upperbounds{S}=\{\chi\in Q \mid \forall s\in S\ldotp s \sqsubseteq\chi\} \quad \textrm{(upper bounds)}
\]\[ %
\lowerbounds{S}=\{\chi\in Q \mid \forall s\in S\ldotp \chi\sqsubseteq s \} \quad \textrm{(lower bounds)}
\]
\[\downarrow x = \{y \mid y \sqsubseteq x\}\qquad \downarrow S=\{y\mid \exists x\in S\ldotp y\sqsubseteq x\}\]
\newcommand{\mac}[1]{\ensuremath{\mathbb{M}(#1)}}
\begin{definition}[Standard Completion]
The standard, or {(Dedekind-)MacNeille} completion \mac{Q} of a partial order $Q$
is defined as
$\mac{Q}=\{S \subseteq Q \mid \lowerbounds{\upperbounds{S}}=S\}$.
That is, the completion is downward-closed subsets of $Q$ which are the lower bounds of their own upper bounds.
\end{definition}
\begin{proposition}
\label{prop:macneille}
\mac{Q} has several useful properties:
\begin{itemize}
\item \mac{Q} is a complete lattice under the usual set operations, ordered by $\subseteq$.
\item $Q$ embeds into $\mac{Q}$ via $\chi\mapsto {\downarrow{\chi}}$; alternatively, $\downarrow-$ is an order embedding.
\item $\downarrow-$ preserves all meets and joins that already exist in $Q$.
\end{itemize}
\end{proposition}
Note that if $x\sqcup y$ is undefined in $Q$, the join of ${\downarrow x}$ and ${\downarrow y}$ in \mac{Q} is all of $Q$ (\emph{not} set union of the downward closures).
This essentially over-approximates undefined joins with a new synthetic top element.

\subsection{Effect Quantale Completions}
As noted earlier, we must complete not only the order, but also the monoid: effect quantales (and effect systems) intentionally avoid classifying undesirable behaviours, in part by leaving $\rhd$ undefined. We must add a way to classify those behaviours for the Galois connection.

We define a monoid completion operation $\bigstar(Q)$ on an effect quantale $Q$ to add a new element $\bigstar$ if $\rhd$ is partial, extending $\rhd$ to yield $\bigstar$ whenever $\rhd$ was undefined in $Q$. $\sqcup$ is extended such that $\bigstar$ is unordered with any other element. If $Q$ was already total, $\bigstar(Q)=Q$.

\newcommand{\depower}[1]{\ensuremath{\mathbb{P}_{\not\sqsubseteq}(#1)}\xspace}

Then we define the completion $\mathbb{C}(Q)=\mac{\bigstar(Q)}$.\footnote{Note that $\bigstar(Q)$ may not satisfy the effect quantale distributivity laws; distributivity is restored by applying \mac{-} as well.}

\[
\begin{array}{r@{\;}c@{\;}l}
\gamma(-)&:&\mathbb{C}(Q)\rightarrow \power{Q^*}\\
\gamma(S)&=&\bigcup_{\chi\in S}\{T\mid \mathsf{fold}\;{\rhd}\;I\;T\sqsubseteq \chi\}\\
\\
\alpha(-)&:&\power{Q^*}\rightarrow{\mathbb{C}(Q)}\\
\alpha(S)&=& \bigsqcup{\left(\{{\downarrow}(\mathsf{fold}\;{\rhd}\;I\;T) \mid T\in S\}\right)}
\end{array}
\]
Note that if the completion was more than a no-op,
this concretizes $\bigstar$ to the set of invalid traces and $\top$ to the set of all traces, as suggested above,
yielding a Galois connection
\[\langle\power{Q^*},\subseteq\rangle \xleftrightharpoons[\alpha]{\gamma} \langle{\mathbb{C}(Q)},\subseteq\rangle\]
on computations, which we can now extend to the full semantics.
This is in fact a Galois insertion: $\alpha(\gamma(S))=S$.

The abstract operations on this domain are then post-composition of (concrete) effects. And unsurprisingly,
post-composition in $\mathbb{C}(Q)$ is the most precise abstraction of post-composition in the semantics:
for any effect $\chi\in Q$:
\[ X\mapsto X\rhd \chi \equiv X\mapsto\alpha(\gamma(X)\rhd\chi) \]

\section{Effect Systems as Abstract Interpretations}
We are now equipped to fully model an abstract effect system as a construction of an abstract interpretation in Figure \ref{fig:eff-as-absint}.
Following \citet{cousot1997types}, we separately treat type environments and types as domains, and construct the domain for a typing judgment from these (plus the domain of effects).
The most interesting in this regard is the treatment of the function type and its latent effect;
notice that the function type does not constrain the execution trace that produces the function,
only the behaviour of the function itself.
We offer the following interpretations for types as an abstract domain,
assuming
\[
\begin{array}{r@{\,}c@{\,}l}
\tau\in\mathsf{Type} & ::= & \mathsf{unit} \mid \tau\xrightarrow\tau'\\
\Gamma\in\mathsf{TypeEnvironment} & ::= & \cdot \mid \Gamma,x:\tau\\
\mathsf{Typing}\hspace{-0.2em} &=& \hspace{-0.2em}\mathsf{TypeEnvironment}\times\mathsf{Type}\times \mathbb{C}(Q)\\
\mathsf{ProgramType}\hspace{-0.2em}&=&\hspace{-0.2em}\power{\mathsf{Typing}}
\end{array}
\]
While standard in treatments of type systems as abstract interpretations, we note for type-systems-oriented readers that the distinction between \textsf{Typing} and \textsf{ProgramType} captures the distinction between a chosen typing for a program expression and the space of all valid typings for a program expression. For example, a lone variable $x$ has a separate \textsf{Typing} for each type context that assigns a type to $x$.\footnote{This consequently means that $\emptyset$ is the greatest program type, and $\mathsf{Typing}$ is the least program type.}
\begin{figure}
\[
\begin{array}{r@{\;}c@{\;}l}
\gamma(-)&:&\mathsf{Type}\rightarrow\power{\mathbb{V}}\\
\gamma(\mathsf{unit})&=& \{\mathsf{tt}\}\\
\gamma(\tau\xrightarrow\chi\tau')&=&\{\langle f,T\rangle \mid \forall v\in\gamma(\tau)\ldotp f\;v\in(\gamma(\tau')\times\gamma(\chi)) \}\\
\\[-0.5em]
\gamma(-)&:&\mathsf{TypeEnvironment}\rightarrow\power{\mathbb{R}}\\
\gamma(\epsilon)&=&\{\emptyset\}~\textrm{(i.e., the singleton set with empty function)}\\
\gamma(\Gamma,x:\tau)&=& \{ R\in\mathbb{R} \mid \exists R'\in\gamma(\Gamma) \ldotp
  \exists v\in\gamma(\tau)\ldotp R=R[x\mapsto v]\}\\
\\[-0.5em]
\gamma(-)&:& \mathsf{Typing} \rightarrow\power{\mathbb{V}\times Q^*}\\
\gamma(\langle\Gamma,\tau,\chi\rangle)&=& \{ \phi\in\mathbb{S} \mid \forall R\in\gamma(\Gamma)\ldotp \phi(R)\in \gamma(\tau)\times\gamma(\chi) \} \\
\\[-0.5em]
\gamma(-)&:&\mathsf{ProgramType}\rightarrow\power{\mathbb{S}}\\
\gamma(\overline{T})&=&\bigcap\{\gamma(\langle\Gamma,\tau,\chi\rangle)\mid \langle\Gamma,\tau,\chi\rangle\in\overline{T}\}
\end{array}
\]
\vspace{-1em}
\caption{Concretizing parts of a type-and-effect system.}
\label{fig:eff-as-absint}
\end{figure}

$\gamma$ preserves meets between these two complete lattices: for any collection $\overline{T}_{i\in\Delta}$:
\[ \gamma\left(\bigcup_{i\in\Delta}\overline{T}_i\right) = \bigcap_{i\in\Delta}\gamma(\overline{T_i}) \]
(which is meet preservation because the ordering on \textsf{ProgramType} is $\supseteq$, not $\subseteq$).
Thus
there is a Galois connection
\[ \langle\power{\mathbb{S}},\subseteq\rangle\xleftrightharpoons[\alpha]{\gamma}\langle\mathsf{ProgramType},\supseteq\rangle  \]
with $\alpha$ the unique lower adjoint defined as:
\[\alpha(S)=\bigcup\{\overline{T} \mid S\subseteq\gamma(\overline{T})\}\]
i.e., a set of program semantics is abstracted to the union of all program types whose concretizations include the semantics in question.
This makes \textsf{ProgramType}s a sound abstraction of program semantics, and the typings $\mathbb{T}\llbracket e\rrbracket$ of a program expression $e$ are
\[
\mathbb{T}\llbracket e\rrbracket= \alpha(\{\csem{e}\})
\]
We can now prove soundness of the rules in Figure \ref{fig:lang}, and type safety, given the redefinition of typing judgments as $\Gamma\vdash e : \tau \mid \chi \equiv \langle\Gamma,\tau,\downarrow\chi\rangle\in\mathbb{T}\llbracket e\rrbracket\land\chi\in Q$.
\begin{lemma}[Type Safety]
\label{lem:type_safety}
If $\Gamma\vdash e : \tau \mid \chi$, then
\[\csem{e}\subseteq \gamma(\langle\Gamma,\tau,\downarrow\chi\rangle) \]
or stated equivalently,
\[ \forall\phi\in\csem{e}\ldotp \forall R\in\gamma(\Gamma)\ldotp \phi(R)\in\gamma(\tau)\times\gamma(\downarrow\chi) \]
\end{lemma}
\begin{proof}
Follows from the definition of concretization.
\end{proof}
\begin{lemma}[Type Rule Soundness]
\label{lem:rule_soundness}
Each of Figure \ref{fig:lang}'s type rules (with $\chi,\chi'$ ranging over $Q$, not $\mathbb{C}(Q)$) is a valid inference.
\end{lemma}
\begin{proof}
Rule by rule, the soundness follows directly from the membership interpretation of the judgment.
\end{proof}
Note that the restriction of the effect positions in the type rules to elements of $Q$ rather than its completion retains the partial behaviour of the original type rules.

For reasons of space we do not pursue the more traditional abstract interpretation goal of \emph{deriving} typing rules from the abstraction. However the results above are already interesting for multiple reasons.
First, the proofs apply to a very general class of effect systems, yielding both a more general connection between effect systems and abstract interpretation than has previously been established.
Second, these results resurface a seemingly overlooked connection between the abstract interpretation approach to type systems
and the modern approach of what is now called the \emph{logical approach to type soundness}~\cite{timany2024logical}, where types are given denotations as formulae in a program logic. But earlier work in this space~\cite{gordon2012uniqueness} used an approach~\cite{dinsdale2013views} where types were explicitly given denotations as \emph{sets}, which were then treated logically in the same way as in abstract interpretation~\cite{cousot2024calculational}.

\paragraph{A \textsc{Crit} Example}
Let us consider two \textsc{Crit} programs, using \texttt{acq()} as shorthand for emitting the lock acquisition effect ($\mathsf{ev}(\textsf{locking})$), and \texttt{rel()} for release ($\mathsf{ev}(\textsf{unlocking})$).
Let $e_1=\texttt{acq();rel();}$, which Fig.~\ref{fig:lang} gives effect $\mathsf{entrant}$ ($=\mathsf{locking}\rhd\mathsf{unlocking}$) and
type $\mathsf{unit}$: $\cdot\vdash e_1;e_2 : \mathsf{unit} \mid \mathsf{entrant}$.
Thus Lemma \ref{lem:type_safety} implies a corresponding compatible typing:
\[\csem{e_1}{=}\{\Lambda R{\ldotp}\{\langle\mathsf{tt},\mathsf{locking}\cdot\mathsf{unlocking}\rangle\}\}{\subseteq}\gamma(\langle\cdot,\mathsf{unit},{\downarrow}\mathsf{entrant}\rangle)\]
Consider $e_2=\texttt{acq();acq();}$, which is ill-typed ($\mathsf{locking}\rhd\mathsf{locking}$ is undefined in Figure \ref{fig:crit-effects})
but whose semantics are classified by a typing:
\[\csem{e_2}{=}\{\Lambda R\ldotp\{\langle\mathsf{tt},\mathsf{locking}\cdot\mathsf{locking}\rangle\}\}\subseteq\gamma(\langle\cdot,\mathsf{unit},Q\cup\{\bigstar\}\rangle)\]
$\bigstar$ must occur in any typing of $e_2$ in order to classify traces like $\mathsf{locking}\cdot\mathsf{locking}$ that the effect system locally rejects.
\looseness=-1

\section{Iteration}
Most effect quantales of interest are also \emph{iterable}, adding an operator $(-)^*$ which is useful for giving effects to recursive programs or unbounded finite repetition:
\begin{definition}[Iterable Effect Quantale]
\label{def:iterable}
An \emph{iterable} effect quantale $Q$ carries an additional operator $(-)^*$ (called iteration) which is
monotone ($x\sqsubseteq y\Rightarrow x^*\sqsubseteq y^*$),
idempotent ($(x^*)^*=x^*$),
extensive ($x\sqsubseteq x^*$),
non-empty ($I\sqsubseteq x^*$),
and
foldable ($x^*\rhd x^*\sqsubseteq x^*$).
\end{definition}
This is used directly in typing judgments for
certain recursive effects~\cite{gordon2020lifting}, and
explicit looping constructs, as in:
\[
\inferrule*[left=T-While]{
  \Gamma\vdash e_c : \mathsf{boolean} \mid \chi_c\\
  \Gamma\vdash e_b : \mathsf{unit} \mid \chi_b
}{
  \Gamma\vdash \mathsf{while}(e_c)\{e_b\} : \mathsf{unit} \mid \chi_c\rhd(\chi_b\rhd\chi_c)^*
}
\]

The first three properties make iteration a closure operator.
Iteration in an effect quantale is subtle: unlike related structures such as Kleene Algebras~\cite{KOZEN1994366}, the effect quantale axioms do not guarantee a unique least fixed point of the functors $F_a(X)=1\sqcup(a\rhd X)$, and the axioms permit structures with infinite descending chains of increasingly lower fixed points but no \emph{least} fixed point.
Thus in general effect quantales require a \emph{choice} of iteration operator, even if in many cases a unique most precise (least fixed point) exists~\cite{gordon2021polymorphic}.
Effect quantales which do possess an optimally precise iteration are called \emph{principally iterable}.
But in general the details of iteration depend on the specific analysis.
But we have developed enough structure to recover an equivalent definition of principal iteration (when defined), and to compare the most common reusable method of computing iteration to computation of fixed points in abstract interpretation.

\subsection{Principal Iteration via Abstraction}
Iteration in the effect system is intended to abstract the finite iteration of the concrete behaviours, so it is natural to consider
concrete iteration, and then build the most precise abstraction of that operator.
Because effects model, essentially, formal languages over $Q$, it is natural to consider
\[
\begin{array}{r@{\;}c@{\;}l}
(-)^*&:&\mathbb{C}(Q) \rightarrow \mathbb{C}(Q)\\
\chi^* &\equiv& \alpha(\gamma(\chi)^*)
\end{array} \]
where the ${}^*$ on the right is the Kleene iteration operator on formal languages.
This construction validates the expected axioms, because we are working with a Galois insertion
and the abstraction of an idempotent formal language (one where $L\cdot L=L$) will yield idempotent elements.

This is principal (optimally precise) for $\mathbb{C}(Q)$.
However the results may not correspond to elements of  $Q$, since the completion contains elements
which are not embedings of effects.
But we may recover the definition of principal effect iteration, when defined, by restricting to the case where the result is the downward closure of a single element:
\[
\begin{array}{r@{\;}c@{\;}l}
(-)^*&:&Q \rightharpoonup Q\\
\chi^* &\equiv& \alpha(\gamma(\downarrow\chi)^*)~\textrm{if}~\exists \chi'\in Q\ldotp  \downarrow\chi'=\alpha(\gamma(\downarrow\chi)^*)
\end{array}
\]
This leaves unresolved the questions of how we may effectively compute the above,
and what to do with non-principally iterable effect quantales.
The former, similar to different abstract domains, depends on the specific analysis.
The second should correspond to a widening operator, but we lack space to explore the details.

\subsection{Iteration in Finite Effect Quantales}
\label{sec:finite}
In the case of \emph{finite} effect quantales, the least fixed point can be directly computed by iterative approximation:
\begin{proposition}[Finite Effect Iteration~\textrm{\cite[Prop.~5.11]{gordon2021polymorphic}}]
\label{prop:finite_iter}
A finite effect quantale $Q$ is principally iterable, with iteration $\chi^*$ given by iteratively computing successive powers of $(\chi\sqcup I)$ until a subidempotent power $n$ is reached, i.e., one where
\[(\chi\sqcup I)^n\rhd(\chi\sqcup I)^n\sqsubseteq (\chi\sqcup I)^n\]
If any power of $(\chi\sqcup I)$ is undefined, then $\chi^*$ is undefined.
\end{proposition}
\begin{proof}[Proof Sketch]
See Gordon~\cite[Prop.~5.11]{gordon2021polymorphic} for full details,
but briefly, $\chi\sqcup I$ ensures the result is extensive and non-empty,
and the iterative powers eventually find an idempotent element, which must exist in any finite semigroup~\cite{almeida2009representation}.
\end{proof}
Iteration models the limit of finite iteration of a loop body with effect $\chi$.
Gordon's description is based on the abstract characterization of (principal) iteration of an effect $\chi$ as producing the least idempotent element $\chi'$ ($\chi'\rhd\chi'=\chi'$) greater than both $\chi$ and $I$.
This turns out to be equivalent in not only result, but method, to how the computation would unfold by directly computing the least fixed point in abstract interpretation.
The iterates in Gordon's approach are
\[(\chi\sqcup I)\sqsubseteq (\chi\sqcup I)^2 \sqsubseteq (\chi\sqcup I)^3 \sqsubseteq \ldots\]
which by distributivity are equivalent to:
\[I\sqcup\chi \sqsubseteq I\sqcup\chi\sqcup\chi^2\sqsubseteq  I\sqcup\chi\sqcup\chi^2\sqcup\chi^3 \subseteq\ldots\]
i.e., the sequence given by $S(i)=\bigsqcup_{n\le i}\chi^i$.
This sequence is of course nearly the same as what results from iterating the standard fixed point functor $F(X)=I\sqcup(\chi\rhd X)$, where the iterates (starting from $I$) are
\[ F^0(I)=I\sqsubseteq I\sqcup \chi\sqsubseteq I\sqcup\chi\sqcup\chi^2\sqsubseteq\ldots\]
So $F^i(I)=S(i)$ for all $i$. Gordon simply skips the first iterate.

\section{Related Work}
As discussed in \S\ref{sec:intro}, many individual effect systems have individually been cast as abstract interpretations.
Effect quantales are a relaxation of \emph{quantales} (total complete-lattice-ordered monoids), which serve as models of linear logic~\cite{yetter1990quantales}, and have been used to ensure completeness properties in abstract interpretation~\cite{giacobazzi1998building,giacobazzi2005making}. However effect quantales are weaker (Fig.~\ref{fig:crit-effects} is not a quantale),  
and our results are the first systematic relationship between a broad class of effect systems (effect quantales) and abstract interpretations.
\looseness=-1

Unlike some effect systems~\cite{nicola2025abstract,sekiyama2023temporal,sekiyama2025algebraic},
this machinery does not treat dependent effects --- those where effects may depend on pure values, even as singletons~\cite{gordon2021polymorphic} (for mere identity tracking, such as for locks~\cite{boyapati02,suenaga2008type}).
Adapting the approach of \citet{nicola2025abstract} into our framework would mean concretizing the type and effect of a typing \emph{jointly}, which would also be rich enough to treat cases like \citet{skalka2020types} where an effect can depend on a returned value.
But more general treatments where arbitrary values may be mentioned in effects requires extending the Galois connection further, with concretization of effects carrying a signature
\[\gamma(-):\mathcal{C}(Q)\rightarrow\power{\mathbb{R}\rightarrow\power{\mathbb{V}\times Q^*}}\]

More subtly, this machinery does not treat type-dependency in effects, where effects must be mutually defined with types. This is required for effect systems dealing with continuations~\cite{gordon2020lifting,sekiyama2023temporal}, where effects must track the types of values returned.
This would be a more substantial investigation, as abstracting mutual definition of types and effects is simply a delicate balancing act; see \citet{gordon2021polymorphic} for a syntactic example.
\looseness=-1

This approach also does not yet extend to polymorphism over effects (or types), though this plays a major role in practical expressive effect systems~\cite{vanDooren2005,nielson1994constraints,gordon2013java,BIRKEDAL2001299} --- or more precisely, a gap remains between extending the results above (which would be adequate for soundness) to a working inference tool.
\citet{cousot1997types}, whose approach we followed here, applied his technique to ML-style prenex type polymorphism.
The key difficulty is that while inferred \emph{types} are usually resolved by term unification (in the sense of term algebras, not program terms),
effect systems satisfy a range of algebraic equations, rather than merely structural congurences.
Most effect systems which have explicitly cast themselves as abstract interpretations have not included effect polymorphism, so all calculations can be done directly if the effect operations (join, etc.) are computable.
\citet{Skalka2008} includes a form of effect polymorphism where effect variables are replaced by references to which methods were called, essentially adding new atomic effects to their trace effects,
and still working in a fragment of a free structure over atomic events.
\emph{Commutative} effect systems can extend the approaches of \citet{jouvelot1991algebraic}, or in some cases (such as Flix~\cite{madsen2023fast}) boolean unification~\cite{kennedy1996type}.
But in general effects with effect variables are essentially non-commutative polynomials, making the problem much more difficult.

Resolving all of these limitations, systematically for effect systems writ large, is worthwhile future work.

\section{Conclusion}
This paper gives the first systematic formal connection between (flow-sensitive) effect systems and abstract interpretations.
By following the approach of \citet{cousot1997types} extended with a completion of effect quantales into an abstract domain,
we were able to prove soundness of a type-and-effect system as a fragment of an abstract interpretation.
This provides a basis that can be used to develop abstract interpreters that analyze programs according to an effect system, potentially yielding
a productive approach to effect inference.

\ifPublic
\begin{acks}
This work was supported by NSF Award \#CCF-2007582.
\end{acks}
\fi

\bibliographystyle{ACM-Reference-Format}
\bibliography{effects,abstractinterp}


\begin{thebibliography}{60}


\ifx \showCODEN    \undefined \def \showCODEN     #1{\unskip}     \fi
\ifx \showISBNx    \undefined \def \showISBNx     #1{\unskip}     \fi
\ifx \showISBNxiii \undefined \def \showISBNxiii  #1{\unskip}     \fi
\ifx \showISSN     \undefined \def \showISSN      #1{\unskip}     \fi
\ifx \showLCCN     \undefined \def \showLCCN      #1{\unskip}     \fi
\ifx \shownote     \undefined \def \shownote      #1{#1}          \fi
\ifx \showarticletitle \undefined \def \showarticletitle #1{#1}   \fi
\ifx \showURL      \undefined \def \showURL       {\relax}        \fi
\providecommand\bibfield[2]{#2}
\providecommand\bibinfo[2]{#2}
\providecommand\natexlab[1]{#1}
\providecommand\showeprint[2][]{arXiv:#2}
\makeatletter
\@ifundefined{NAT@parse@date}{}{\let\NAT@parse@date@orig\NAT@parse@date}
\@ifundefined{NAT@parse@date}{}{\def\NAT@parse@date#1#2#3#4#5#6@@{\NAT@parse@date@orig#1#2#3#4#5#6@@\def\NAT@tempyear{0000}\def\NAT@tempexlab{{?}}\ifx\NAT@year\NAT@tempyear\ifx\NAT@exlab\NAT@tempexlab\def\NAT@date{[n.\,d.]}\else\edef\NAT@date{[n.\,d.]\NAT@exlab}\fi\fi}}
\makeatother

\bibitem[Almeida et~al\mbox{.}(2009)]%
        {almeida2009representation}
\bibfield{author}{\bibinfo{person}{Jorge Almeida}, \bibinfo{person}{Stuart
  Margolis}, \bibinfo{person}{Benjamin Steinberg}, {and}
  \bibinfo{person}{Mikhail Volkov}.} \bibinfo{year}{2009}\natexlab{}.
\newblock \showarticletitle{Representation theory of finite semigroups,
  semigroup radicals and formal language theory}.
\newblock \bibinfo{journal}{\emph{Trans. Amer. Math. Soc.}}
  \bibinfo{volume}{361}, \bibinfo{number}{3} (\bibinfo{year}{2009}),
  \bibinfo{pages}{1429--1461}.
\newblock
\href{https://doi.org/10.1090/S0002-9947-08-04712-0}{doi:\nolinkurl{10.1090/S0002-9947-08-04712-0}}


\bibitem[America and Rutten(1989)]%
        {america1989solving}
\bibfield{author}{\bibinfo{person}{Pierre America} {and} \bibinfo{person}{Jan
  Rutten}.} \bibinfo{year}{1989}\natexlab{}.
\newblock \showarticletitle{Solving reflexive domain equations in a category of
  complete metric spaces}.
\newblock \bibinfo{journal}{\emph{J. Comput. System Sci.}}
  \bibinfo{volume}{39}, \bibinfo{number}{3} (\bibinfo{year}{1989}),
  \bibinfo{pages}{343--375}.
\newblock
\href{https://doi.org/10.1016/0022-0000(89)90027-5}{doi:\nolinkurl{10.1016/0022-0000(89)90027-5}}


\bibitem[Amtoft et~al\mbox{.}(1999)]%
        {amtoft1999}
\bibfield{author}{\bibinfo{person}{Torben Amtoft}, \bibinfo{person}{Flemming
  Nielson}, {and} \bibinfo{person}{Hanne~Riis Nielson}.}
  \bibinfo{year}{1999}\natexlab{}.
\newblock \bibinfo{booktitle}{\emph{{Type and Effect Systems: Behaviours for
  Concurrency}}}.
\newblock \bibinfo{publisher}{Imperial College Press},
  \bibinfo{address}{London, UK}.
\newblock


\bibitem[Bao et~al\mbox{.}(2021)]%
        {bao2021reachability}
\bibfield{author}{\bibinfo{person}{Yuyan Bao}, \bibinfo{person}{Guannan Wei},
  \bibinfo{person}{Oliver Bra{\v{c}}evac}, \bibinfo{person}{Yuxuan Jiang},
  \bibinfo{person}{Qiyang He}, {and} \bibinfo{person}{Tiark Rompf}.}
  \bibinfo{year}{2021}\natexlab{}.
\newblock \showarticletitle{Reachability types: tracking aliasing and
  separation in higher-order functional programs}.
\newblock \bibinfo{journal}{\emph{Proceedings of the ACM on Programming
  Languages}} \bibinfo{volume}{5}, \bibinfo{number}{OOPSLA}
  (\bibinfo{year}{2021}), \bibinfo{pages}{1--32}.
\newblock
\href{https://doi.org/10.1145/3485516}{doi:\nolinkurl{10.1145/3485516}}


\bibitem[Bauer and Pretnar(2014)]%
        {bauer2014effect}
\bibfield{author}{\bibinfo{person}{Andrej Bauer} {and} \bibinfo{person}{Matija
  Pretnar}.} \bibinfo{year}{2014}\natexlab{}.
\newblock \showarticletitle{An effect system for algebraic effects and
  handlers}.
\newblock \bibinfo{journal}{\emph{Logical methods in computer science}}
  \bibinfo{volume}{10} (\bibinfo{year}{2014}).
\newblock
\href{https://doi.org/10.2168/LMCS-10(4:9)2014}{doi:\nolinkurl{10.2168/LMCS-10(4:9)2014}}


\bibitem[Benton and Buchlovsky(2007)]%
        {benton2007exceptions}
\bibfield{author}{\bibinfo{person}{Nick Benton} {and} \bibinfo{person}{Peter
  Buchlovsky}.} \bibinfo{year}{2007}\natexlab{}.
\newblock \showarticletitle{{Semantics of an Effect Analysis for Exceptions}}.
  In \bibinfo{booktitle}{\emph{{TLDI}}}.
\newblock
\href{https://doi.org/10.1145/1190315.1190320}{doi:\nolinkurl{10.1145/1190315.1190320}}


\bibitem[Birkedal et~al\mbox{.}(2010)]%
        {birkedal2010category}
\bibfield{author}{\bibinfo{person}{Lars Birkedal}, \bibinfo{person}{Kristian
  St{\o}vring}, {and} \bibinfo{person}{Jacob Thamsborg}.}
  \bibinfo{year}{2010}\natexlab{}.
\newblock \showarticletitle{The category-theoretic solution of recursive
  metric-space equations}.
\newblock \bibinfo{journal}{\emph{Theoretical Computer Science}}
  \bibinfo{volume}{411}, \bibinfo{number}{47} (\bibinfo{year}{2010}),
  \bibinfo{pages}{4102--4122}.
\newblock
\href{https://doi.org/10.1016/j.tcs.2010.07.010}{doi:\nolinkurl{10.1016/j.tcs.2010.07.010}}


\bibitem[Birkedal and Tofte(2001)]%
        {BIRKEDAL2001299}
\bibfield{author}{\bibinfo{person}{Lars Birkedal} {and} \bibinfo{person}{Mads
  Tofte}.} \bibinfo{year}{2001}\natexlab{}.
\newblock \showarticletitle{A constraint-based region inference algorithm}.
\newblock \bibinfo{journal}{\emph{Theoretical Computer Science}}
  \bibinfo{volume}{258}, \bibinfo{number}{1} (\bibinfo{year}{2001}),
  \bibinfo{pages}{299 -- 392}.
\newblock
\showISSN{0304-3975}
\href{https://doi.org/10.1016/S0304-3975(00)00025-6}{doi:\nolinkurl{10.1016/S0304-3975(00)00025-6}}


\bibitem[Bocchino et~al\mbox{.}(2009)]%
        {bocchino09}
\bibfield{author}{\bibinfo{person}{Robert~L. Bocchino, Jr.},
  \bibinfo{person}{Vikram~S. Adve}, \bibinfo{person}{Danny Dig},
  \bibinfo{person}{Sarita~V. Adve}, \bibinfo{person}{Stephen Heumann},
  \bibinfo{person}{Rakesh Komuravelli}, \bibinfo{person}{Jeffrey Overbey},
  \bibinfo{person}{Patrick Simmons}, \bibinfo{person}{Hyojin Sung}, {and}
  \bibinfo{person}{Mohsen Vakilian}.} \bibinfo{year}{2009}\natexlab{}.
\newblock \showarticletitle{{A Type and Effect System for Deterministic
  Parallel Java}}. In \bibinfo{booktitle}{\emph{OOPSLA}}.
\newblock
\href{https://doi.org/10.1145/1640089.1640097}{doi:\nolinkurl{10.1145/1640089.1640097}}


\bibitem[Boyapati et~al\mbox{.}(2002)]%
        {boyapati02}
\bibfield{author}{\bibinfo{person}{Chandrasekhar Boyapati},
  \bibinfo{person}{Robert Lee}, {and} \bibinfo{person}{Martin Rinard}.}
  \bibinfo{year}{2002}\natexlab{}.
\newblock \showarticletitle{{Ownership Types for Safe Programming: Preventing
  Data Races and Deadlocks}}. In \bibinfo{booktitle}{\emph{{OOPSLA}}}.
\newblock
\href{https://doi.org/10.1145/582419.582440}{doi:\nolinkurl{10.1145/582419.582440}}


\bibitem[Cousot(1997)]%
        {cousot1997types}
\bibfield{author}{\bibinfo{person}{Patrick Cousot}.}
  \bibinfo{year}{1997}\natexlab{}.
\newblock \showarticletitle{Types as abstract interpretations}. In
  \bibinfo{booktitle}{\emph{Proceedings of the symposium on Principles of
  programming languages}}.
\newblock
\href{https://doi.org/10.1145/263699.263744}{doi:\nolinkurl{10.1145/263699.263744}}


\bibitem[Cousot(2024)]%
        {cousot2024calculational}
\bibfield{author}{\bibinfo{person}{Patrick Cousot}.}
  \bibinfo{year}{2024}\natexlab{}.
\newblock \showarticletitle{Calculational design of [in] correctness
  transformational program logics by abstract interpretation}.
\newblock \bibinfo{journal}{\emph{Proceedings of the ACM on Programming
  Languages}} \bibinfo{volume}{8}, \bibinfo{number}{POPL}
  (\bibinfo{year}{2024}), \bibinfo{pages}{175--208}.
\newblock
\href{https://doi.org/10.1145/3632849}{doi:\nolinkurl{10.1145/3632849}}


\bibitem[Cousot and Cousot(1979)]%
        {cousot1979systematic}
\bibfield{author}{\bibinfo{person}{Patrick Cousot} {and}
  \bibinfo{person}{Radhia Cousot}.} \bibinfo{year}{1979}\natexlab{}.
\newblock \showarticletitle{Systematic design of program analysis frameworks}.
  In \bibinfo{booktitle}{\emph{Proceedings of the 6th ACM SIGACT-SIGPLAN
  symposium on Principles of programming languages}}.
  \bibinfo{pages}{269--282}.
\newblock
\href{https://doi.org/10.1145/567752.567778}{doi:\nolinkurl{10.1145/567752.567778}}


\bibitem[Cousot and Cousot(1992)]%
        {cousot1992abstract}
\bibfield{author}{\bibinfo{person}{Patrick Cousot} {and}
  \bibinfo{person}{Radhia Cousot}.} \bibinfo{year}{1992}\natexlab{}.
\newblock \showarticletitle{Abstract interpretation frameworks}.
\newblock \bibinfo{journal}{\emph{Journal of logic and computation}}
  \bibinfo{volume}{2}, \bibinfo{number}{4} (\bibinfo{year}{1992}),
  \bibinfo{pages}{511--547}.
\newblock
\href{https://doi.org/10.1093/logcom/2.4.511}{doi:\nolinkurl{10.1093/logcom/2.4.511}}


\bibitem[Dagand et~al\mbox{.}(2018)]%
        {dagand2018foundations}
\bibfield{author}{\bibinfo{person}{Pierre-{\'E}variste Dagand},
  \bibinfo{person}{Nicolas Tabareau}, {and} \bibinfo{person}{{\'E}ric Tanter}.}
  \bibinfo{year}{2018}\natexlab{}.
\newblock \showarticletitle{Foundations of dependent interoperability}.
\newblock \bibinfo{journal}{\emph{Journal of Functional Programming}}
  \bibinfo{volume}{28} (\bibinfo{year}{2018}), \bibinfo{pages}{e9}.
\newblock
\href{https://doi.org/10.1017/S0956796818000011}{doi:\nolinkurl{10.1017/S0956796818000011}}


\bibitem[Dinsdale-Young et~al\mbox{.}(2013)]%
        {dinsdale2013views}
\bibfield{author}{\bibinfo{person}{Thomas Dinsdale-Young},
  \bibinfo{person}{Lars Birkedal}, \bibinfo{person}{Philippa Gardner},
  \bibinfo{person}{Matthew Parkinson}, {and} \bibinfo{person}{Hongseok Yang}.}
  \bibinfo{year}{2013}\natexlab{}.
\newblock \showarticletitle{Views: compositional reasoning for concurrent
  programs}. In \bibinfo{booktitle}{\emph{Proceedings of the symposium on
  principles of programming languages}}.
\newblock
\href{https://doi.org/10.1145/2429069.2429104}{doi:\nolinkurl{10.1145/2429069.2429104}}


\bibitem[Flanagan and Abadi(1999)]%
        {objtyrace99}
\bibfield{author}{\bibinfo{person}{Cormac Flanagan} {and}
  \bibinfo{person}{Mart\'in Abadi}.} \bibinfo{year}{1999}\natexlab{}.
\newblock \showarticletitle{{Object Types against Races}}. In
  \bibinfo{booktitle}{\emph{{CONCUR}}}.
\newblock
\href{https://doi.org/10.1007/3-540-48320-9_21}{doi:\nolinkurl{10.1007/3-540-48320-9_21}}


\bibitem[Flanagan and Freund(2000)]%
        {rccjava00}
\bibfield{author}{\bibinfo{person}{Cormac Flanagan} {and}
  \bibinfo{person}{Stephen~N. Freund}.} \bibinfo{year}{2000}\natexlab{}.
\newblock \showarticletitle{{Type-Based Race Detection for Java}}. In
  \bibinfo{booktitle}{\emph{{PLDI}}}.
\newblock
\href{https://doi.org/10.1145/349299.349328}{doi:\nolinkurl{10.1145/349299.349328}}


\bibitem[Flanagan and Qadeer(2003)]%
        {flanagan2003atomicity}
\bibfield{author}{\bibinfo{person}{Cormac Flanagan} {and} \bibinfo{person}{Shaz
  Qadeer}.} \bibinfo{year}{2003}\natexlab{}.
\newblock \showarticletitle{A Type and Effect System for Atomicity}. In
  \bibinfo{booktitle}{\emph{Proceedings of the ACM SIGPLAN 2003 Conference on
  Programming Language Design and Implementation}} \emph{(\bibinfo{series}{PLDI
  '03})}. \bibinfo{publisher}{ACM}, \bibinfo{pages}{338--349}.
\newblock
\showISBNx{1-58113-662-5}
\href{https://doi.org/10.1145/781131.781169}{doi:\nolinkurl{10.1145/781131.781169}}


\bibitem[Giacobazzi et~al\mbox{.}(1998)]%
        {giacobazzi1998building}
\bibfield{author}{\bibinfo{person}{Roberto Giacobazzi},
  \bibinfo{person}{Francesco Ranzato}, {and} \bibinfo{person}{Francesca
  Scozzari}.} \bibinfo{year}{1998}\natexlab{}.
\newblock \showarticletitle{Building complete abstract interpretations in a
  linear logic-based setting}. In \bibinfo{booktitle}{\emph{International
  Static Analysis Symposium}}. Springer, \bibinfo{pages}{215--229}.
\newblock
\href{https://doi.org/10.1007/3-540-49727-7_13}{doi:\nolinkurl{10.1007/3-540-49727-7_13}}


\bibitem[Giacobazzi et~al\mbox{.}(2005)]%
        {giacobazzi2005making}
\bibfield{author}{\bibinfo{person}{Roberto Giacobazzi},
  \bibinfo{person}{Francesco Ranzato}, {and} \bibinfo{person}{Francesca
  Scozzari}.} \bibinfo{year}{2005}\natexlab{}.
\newblock \showarticletitle{Making abstract domains condensing}.
\newblock \bibinfo{journal}{\emph{ACM Transactions on Computational Logic
  (TOCL)}} \bibinfo{volume}{6}, \bibinfo{number}{1} (\bibinfo{year}{2005}),
  \bibinfo{pages}{33--60}.
\newblock
\href{https://doi.org/10.1145/1042038.1042040}{doi:\nolinkurl{10.1145/1042038.1042040}}


\bibitem[Gifford and Lucassen(1986)]%
        {gifford86}
\bibfield{author}{\bibinfo{person}{David~K. Gifford} {and}
  \bibinfo{person}{John~M. Lucassen}.} \bibinfo{year}{1986}\natexlab{}.
\newblock \showarticletitle{{Integrating Functional and Imperative
  Programming}}. In \bibinfo{booktitle}{\emph{Proceedings of the ACM Conference
  on LISP and Functional Programming}}.
\newblock
\href{https://doi.org/10.1145/319838.319848}{doi:\nolinkurl{10.1145/319838.319848}}


\bibitem[Gordon(2017)]%
        {gordon2017generic}
\bibfield{author}{\bibinfo{person}{Colin~S Gordon}.}
  \bibinfo{year}{2017}\natexlab{}.
\newblock \showarticletitle{A Generic Approach to Flow-Sensitive Polymorphic
  Effects}. In \bibinfo{booktitle}{\emph{31st European Conference on
  Object-Oriented Programming (ECOOP 2017)}}.
\newblock
\href{https://doi.org/10.4230/LIPIcs.ECOOP.2017.13}{doi:\nolinkurl{10.4230/LIPIcs.ECOOP.2017.13}}


\bibitem[Gordon(2020)]%
        {gordon2020lifting}
\bibfield{author}{\bibinfo{person}{Colin~S Gordon}.}
  \bibinfo{year}{2020}\natexlab{}.
\newblock \showarticletitle{Lifting Sequential Effects to Control Operators}.
  In \bibinfo{booktitle}{\emph{34th European Conference on Object-Oriented
  Programming (ECOOP 2020)}}.
\newblock
\href{https://doi.org/10.4230/LIPIcs.ECOOP.2020.23}{doi:\nolinkurl{10.4230/LIPIcs.ECOOP.2020.23}}


\bibitem[Gordon(2021)]%
        {gordon2021polymorphic}
\bibfield{author}{\bibinfo{person}{Colin~S. Gordon}.}
  \bibinfo{year}{2021}\natexlab{}.
\newblock \showarticletitle{Polymorphic Iterable Sequential Effect Systems}.
\newblock \bibinfo{journal}{\emph{ACM Transactions on Programming Languages and
  Systems (TOPLAS)}} (\bibinfo{year}{2021}).
\newblock
\href{https://doi.org/10.1145/3450272}{doi:\nolinkurl{10.1145/3450272}}


\bibitem[Gordon et~al\mbox{.}(2013)]%
        {gordon2013java}
\bibfield{author}{\bibinfo{person}{Colin~S Gordon}, \bibinfo{person}{Werner
  Dietl}, \bibinfo{person}{Michael~D Ernst}, {and} \bibinfo{person}{Dan
  Grossman}.} \bibinfo{year}{2013}\natexlab{}.
\newblock \showarticletitle{JavaUI: effects for controlling UI object access}.
  In \bibinfo{booktitle}{\emph{European Conference on Object-Oriented
  Programming}}. Springer, \bibinfo{pages}{179--204}.
\newblock


\bibitem[Gordon et~al\mbox{.}(2012a)]%
        {gordon2012static}
\bibfield{author}{\bibinfo{person}{Colin~S Gordon}, \bibinfo{person}{Michael~D
  Ernst}, {and} \bibinfo{person}{Dan Grossman}.}
  \bibinfo{year}{2012}\natexlab{a}.
\newblock \showarticletitle{Static lock capabilities for deadlock freedom}. In
  \bibinfo{booktitle}{\emph{Proceedings of the 8th ACM SIGPLAN workshop on
  Types in language design and implementation}}. \bibinfo{pages}{67--78}.
\newblock
\href{https://doi.org/10.1145/2103786.2103796}{doi:\nolinkurl{10.1145/2103786.2103796}}


\bibitem[Gordon et~al\mbox{.}(2012b)]%
        {gordon2012uniqueness}
\bibfield{author}{\bibinfo{person}{Colin~S Gordon}, \bibinfo{person}{Matthew~J
  Parkinson}, \bibinfo{person}{Jared Parsons}, \bibinfo{person}{Aleks
  Bromfield}, {and} \bibinfo{person}{Joe Duffy}.}
  \bibinfo{year}{2012}\natexlab{b}.
\newblock \showarticletitle{Uniqueness and reference immutability for safe
  parallelism}. In \bibinfo{booktitle}{\emph{Proceedings of the ACM
  international conference on Object oriented programming systems languages and
  applications}}.
\newblock
\href{https://doi.org/10.1145/2384616.2384619}{doi:\nolinkurl{10.1145/2384616.2384619}}


\bibitem[Gordon and Yun(2023)]%
        {gordon2023error}
\bibfield{author}{\bibinfo{person}{Colin~S Gordon} {and}
  \bibinfo{person}{Chaewon Yun}.} \bibinfo{year}{2023}\natexlab{}.
\newblock \showarticletitle{Error Localization for Sequential Effect Systems}.
  In \bibinfo{booktitle}{\emph{International Static Analysis Symposium}}.
  Springer, \bibinfo{pages}{343--370}.
\newblock
\href{https://doi.org/10.1007/978-3-031-44245-2_16}{doi:\nolinkurl{10.1007/978-3-031-44245-2_16}}


\bibitem[Hofmann and Chen(2014)]%
        {hofmann2014abstract}
\bibfield{author}{\bibinfo{person}{Martin Hofmann} {and} \bibinfo{person}{Wei
  Chen}.} \bibinfo{year}{2014}\natexlab{}.
\newblock \showarticletitle{Abstract interpretation from B{\"u}chi automata}.
  In \bibinfo{booktitle}{\emph{Proceedings of the Joint Meeting of the
  Twenty-Third EACSL Annual Conference on Computer Science Logic (CSL) and the
  Twenty-Ninth Annual ACM/IEEE Symposium on Logic in Computer Science (LICS)}}.
  \bibinfo{pages}{1--10}.
\newblock
\href{https://doi.org/10.1145/2603088.2603127}{doi:\nolinkurl{10.1145/2603088.2603127}}


\bibitem[Holdermans and Hage(2010)]%
        {holdermans2010polyvariant}
\bibfield{author}{\bibinfo{person}{Stefan Holdermans} {and}
  \bibinfo{person}{Jurriaan Hage}.} \bibinfo{year}{2010}\natexlab{}.
\newblock \showarticletitle{Polyvariant flow analysis with higher-ranked
  polymorphic types and higher-order effect operators}. In
  \bibinfo{booktitle}{\emph{Proceedings of the 15th ACM SIGPLAN international
  conference on Functional programming}}. \bibinfo{pages}{63--74}.
\newblock
\href{https://doi.org/10.1145/1863543.1863554}{doi:\nolinkurl{10.1145/1863543.1863554}}


\bibitem[Iva{\v{s}}kovi{\'c} and Mycroft(2020)]%
        {ivaskovic2020graded}
\bibfield{author}{\bibinfo{person}{Andrej Iva{\v{s}}kovi{\'c}} {and}
  \bibinfo{person}{Alan Mycroft}.} \bibinfo{year}{2020}\natexlab{}.
\newblock \showarticletitle{A graded Monad for deadlock-free concurrency
  (functional pearl)}. In \bibinfo{booktitle}{\emph{Proceedings of the 13th ACM
  SIGPLAN International Symposium on Haskell}}. \bibinfo{pages}{17--30}.
\newblock
\href{https://doi.org/10.1145/3406088.3409024}{doi:\nolinkurl{10.1145/3406088.3409024}}


\bibitem[Iva{\v{s}}ković et~al\mbox{.}(2020)]%
        {ivaskovic2020dataflow}
\bibfield{author}{\bibinfo{person}{Andrej Iva{\v{s}}ković},
  \bibinfo{person}{Alan Mycroft}, {and} \bibinfo{person}{Dominic Orchard}.}
  \bibinfo{year}{2020}\natexlab{}.
\newblock \showarticletitle{{Data-Flow Analyses as Effects and Graded Monads}}.
  In \bibinfo{booktitle}{\emph{5th International Conference on Formal
  Structures for Computation and Deduction (FSCD 2020)}}.
  \bibinfo{pages}{15:1--15:23}.
\newblock
\showISBNx{978-3-95977-155-9}
\showISSN{1868-8969}
\href{https://doi.org/10.4230/LIPIcs.FSCD.2020.15}{doi:\nolinkurl{10.4230/LIPIcs.FSCD.2020.15}}


\bibitem[Jouvelot and Gifford(1991)]%
        {jouvelot1991algebraic}
\bibfield{author}{\bibinfo{person}{Pierre Jouvelot} {and}
  \bibinfo{person}{David Gifford}.} \bibinfo{year}{1991}\natexlab{}.
\newblock \showarticletitle{Algebraic reconstruction of types and effects}. In
  \bibinfo{booktitle}{\emph{Proceedings of the 18th ACM SIGPLAN-SIGACT
  symposium on Principles of programming languages}}.
  \bibinfo{pages}{303--310}.
\newblock
\href{https://doi.org/10.1145/99583.99623}{doi:\nolinkurl{10.1145/99583.99623}}


\bibitem[Kappelmann(2023)]%
        {kappelmann2023transport}
\bibfield{author}{\bibinfo{person}{Kevin Kappelmann}.}
  \bibinfo{year}{2023}\natexlab{}.
\newblock \showarticletitle{Transport via partial galois connections and
  equivalences}. In \bibinfo{booktitle}{\emph{Asian Symposium on Programming
  Languages and Systems}}. Springer, \bibinfo{pages}{225--245}.
\newblock
\href{https://doi.org/10.1007/978-981-99-8311-7_11}{doi:\nolinkurl{10.1007/978-981-99-8311-7_11}}


\bibitem[Kennedy(1996)]%
        {kennedy1996type}
\bibfield{author}{\bibinfo{person}{Andrew~J. Kennedy}.}
  \bibinfo{year}{1996}\natexlab{}.
\newblock \showarticletitle{Type Inference and Equational Theories}.
\newblock \bibinfo{journal}{\emph{Technical Report LIX/RR/96/09, LIX}}
  (\bibinfo{year}{1996}).
\newblock


\bibitem[Koskinen and Terauchi(2014)]%
        {Koskinen14LTR}
\bibfield{author}{\bibinfo{person}{Eric Koskinen} {and} \bibinfo{person}{Tachio
  Terauchi}.} \bibinfo{year}{2014}\natexlab{}.
\newblock \showarticletitle{Local Temporal Reasoning}. In
  \bibinfo{booktitle}{\emph{Proceedings of the Joint EACSL Annual Conference on
  Computer Science Logic (CSL) and the ACM/IEEE Symposium on Logic in Computer
  Science (LICS)}}.
\newblock
\showISBNx{978-1-4503-2886-9}
\href{https://doi.org/10.1145/2603088.2603138}{doi:\nolinkurl{10.1145/2603088.2603138}}


\bibitem[Kozen(1994)]%
        {KOZEN1994366}
\bibfield{author}{\bibinfo{person}{D. Kozen}.} \bibinfo{year}{1994}\natexlab{}.
\newblock \showarticletitle{A Completeness Theorem for Kleene Algebras and the
  Algebra of Regular Events}.
\newblock \bibinfo{journal}{\emph{Information and Computation}}
  \bibinfo{volume}{110}, \bibinfo{number}{2} (\bibinfo{year}{1994}),
  \bibinfo{pages}{366 -- 390}.
\newblock
\showISSN{0890-5401}
\href{https://doi.org/10.1006/inco.1994.1037}{doi:\nolinkurl{10.1006/inco.1994.1037}}


\bibitem[Kozen(1997)]%
        {kozen1997kleene}
\bibfield{author}{\bibinfo{person}{Dexter Kozen}.}
  \bibinfo{year}{1997}\natexlab{}.
\newblock \showarticletitle{Kleene algebra with tests}.
\newblock \bibinfo{journal}{\emph{ACM Transactions on Programming Languages and
  Systems (TOPLAS)}} \bibinfo{volume}{19}, \bibinfo{number}{3}
  (\bibinfo{year}{1997}), \bibinfo{pages}{427--443}.
\newblock
\href{https://doi.org/10.1145/256167.256195}{doi:\nolinkurl{10.1145/256167.256195}}


\bibitem[Lucassen and Gifford(1988)]%
        {lucassen88}
\bibfield{author}{\bibinfo{person}{J.~M. Lucassen} {and} \bibinfo{person}{D.~K.
  Gifford}.} \bibinfo{year}{1988}\natexlab{}.
\newblock \showarticletitle{{Polymorphic Effect Systems}}. In
  \bibinfo{booktitle}{\emph{{Proceedings of the 15th ACM SIGPLAN-SIGACT
  Symposium on Principles of Programming Languages (POPL)}}}.
\newblock
\href{https://doi.org/10.1145/73560.73564}{doi:\nolinkurl{10.1145/73560.73564}}


\bibitem[MacNeille(1937)]%
        {macneille1937partially}
\bibfield{author}{\bibinfo{person}{Holbrook~Mann MacNeille}.}
  \bibinfo{year}{1937}\natexlab{}.
\newblock \showarticletitle{Partially ordered sets}.
\newblock \bibinfo{journal}{\emph{Trans. Amer. Math. Soc.}}
  \bibinfo{volume}{42}, \bibinfo{number}{3} (\bibinfo{year}{1937}).
\newblock
\href{https://doi.org/10.1090/S0002-9947-1937-1501929-X}{doi:\nolinkurl{10.1090/S0002-9947-1937-1501929-X}}


\bibitem[Madsen et~al\mbox{.}(2023)]%
        {madsen2023fast}
\bibfield{author}{\bibinfo{person}{Magnus Madsen}, \bibinfo{person}{Jaco Van~de
  Pol}, {and} \bibinfo{person}{Troels Henriksen}.}
  \bibinfo{year}{2023}\natexlab{}.
\newblock \showarticletitle{Fast and efficient boolean unification for
  Hindley-Milner-style type and effect systems}.
\newblock \bibinfo{journal}{\emph{Proceedings of the ACM on Programming
  Languages}} \bibinfo{volume}{7}, \bibinfo{number}{OOPSLA2}
  (\bibinfo{year}{2023}), \bibinfo{pages}{516--543}.
\newblock
\href{https://doi.org/10.1145/3622816}{doi:\nolinkurl{10.1145/3622816}}


\bibitem[Marino and Millstein(2009)]%
        {marino09}
\bibfield{author}{\bibinfo{person}{Daniel Marino} {and} \bibinfo{person}{Todd
  Millstein}.} \bibinfo{year}{2009}\natexlab{}.
\newblock \showarticletitle{{A Generic Type-and-Effect System}}. In
  \bibinfo{booktitle}{\emph{{TLDI}}}.
\newblock
\href{https://doi.org/10.1145/1481861.1481868}{doi:\nolinkurl{10.1145/1481861.1481868}}


\bibitem[Nicola et~al\mbox{.}(2025)]%
        {nicola2025abstract}
\bibfield{author}{\bibinfo{person}{Mihai Nicola}, \bibinfo{person}{Chaitanya
  Agarwal}, \bibinfo{person}{Eric Koskinen}, {and} \bibinfo{person}{Thomas
  Wies}.} \bibinfo{year}{2025}\natexlab{}.
\newblock \showarticletitle{Abstract Interpretation of Temporal Safety Effects
  of Higher Order Programs}.
\newblock \bibinfo{journal}{\emph{Proceedings of the ACM on Programming
  Languages}} \bibinfo{volume}{9}, \bibinfo{number}{OOPSLA2}
  (\bibinfo{year}{2025}), \bibinfo{pages}{2511--2539}.
\newblock
\href{https://doi.org/10.1145/3763140}{doi:\nolinkurl{10.1145/3763140}}


\bibitem[Nielson and Nielson(1993)]%
        {nielson1993cml}
\bibfield{author}{\bibinfo{person}{Flemming Nielson} {and}
  \bibinfo{person}{Hanne~Riis Nielson}.} \bibinfo{year}{1993}\natexlab{}.
\newblock \showarticletitle{From CML to process algebras}. In
  \bibinfo{booktitle}{\emph{International Conference on Concurrency Theory
  (CONCUR)}}. Springer, \bibinfo{pages}{493--508}.
\newblock
\href{https://doi.org/10.1007/3-540-57208-2_34}{doi:\nolinkurl{10.1007/3-540-57208-2_34}}


\bibitem[Nielson and Nielson(1994)]%
        {nielson1994constraints}
\bibfield{author}{\bibinfo{person}{Flemming Nielson} {and}
  \bibinfo{person}{Hanne~Riis Nielson}.} \bibinfo{year}{1994}\natexlab{}.
\newblock \showarticletitle{Constraints for polymorphic behaviours of
  concurrent ML}. In \bibinfo{booktitle}{\emph{International Conference on
  Constraints in Computational Logics}}. Springer, \bibinfo{pages}{73--88}.
\newblock


\bibitem[Saffrich and Thiemann(2022)]%
        {saffrich2022relating}
\bibfield{author}{\bibinfo{person}{Hannes Saffrich} {and}
  \bibinfo{person}{Peter Thiemann}.} \bibinfo{year}{2022}\natexlab{}.
\newblock \showarticletitle{Relating functional and imperative session types}.
\newblock \bibinfo{journal}{\emph{Logical Methods in Computer Science}}
  \bibinfo{volume}{18} (\bibinfo{year}{2022}).
\newblock
\href{https://doi.org/10.46298/lmcs-18(3:33)2022}{doi:\nolinkurl{10.46298/lmcs-18(3:33)2022}}


\bibitem[Schmidt(2012)]%
        {schmidt2012inverse}
\bibfield{author}{\bibinfo{person}{David~A Schmidt}.}
  \bibinfo{year}{2012}\natexlab{}.
\newblock \showarticletitle{Inverse-limit and topological aspects of abstract
  interpretation}.
\newblock \bibinfo{journal}{\emph{Theoretical Computer Science}}
  \bibinfo{volume}{430} (\bibinfo{year}{2012}), \bibinfo{pages}{23--42}.
\newblock
\href{https://doi.org/10.1016/j.tcs.2012.01.002}{doi:\nolinkurl{10.1016/j.tcs.2012.01.002}}


\bibitem[Sekiyama and Unno(2023)]%
        {sekiyama2023temporal}
\bibfield{author}{\bibinfo{person}{Taro Sekiyama} {and}
  \bibinfo{person}{Hiroshi Unno}.} \bibinfo{year}{2023}\natexlab{}.
\newblock \showarticletitle{Temporal Verification with Answer-Effect
  Modification: Dependent Temporal Type-and-Effect System with Delimited
  Continuations}.
\newblock \bibinfo{journal}{\emph{Proceedings of the ACM on Programming
  Languages}} \bibinfo{volume}{7}, \bibinfo{number}{POPL}
  (\bibinfo{year}{2023}), \bibinfo{pages}{2079--2110}.
\newblock
\href{https://doi.org/10.1145/3571264}{doi:\nolinkurl{10.1145/3571264}}


\bibitem[Sekiyama and Unno(2025)]%
        {sekiyama2025algebraic}
\bibfield{author}{\bibinfo{person}{Taro Sekiyama} {and}
  \bibinfo{person}{Hiroshi Unno}.} \bibinfo{year}{2025}\natexlab{}.
\newblock \showarticletitle{Algebraic Temporal Effects: Temporal Verification
  of Recursively Typed Higher-Order Programs}.
\newblock \bibinfo{journal}{\emph{Proceedings of the ACM on Programming
  Languages}} \bibinfo{volume}{9}, \bibinfo{number}{POPL}
  (\bibinfo{year}{2025}), \bibinfo{pages}{2306--2336}.
\newblock
\href{https://doi.org/10.1145/3704914}{doi:\nolinkurl{10.1145/3704914}}


\bibitem[Skalka(2008)]%
        {Skalka2008}
\bibfield{author}{\bibinfo{person}{Christian Skalka}.}
  \bibinfo{year}{2008}\natexlab{}.
\newblock \showarticletitle{Types and trace effects for object orientation}.
\newblock \bibinfo{journal}{\emph{Higher-Order and Symbolic Computation}}
  \bibinfo{volume}{21}, \bibinfo{number}{3} (\bibinfo{year}{2008}),
  \bibinfo{pages}{239--282}.
\newblock
\href{https://doi.org/10.1007/s10990-008-9032-6}{doi:\nolinkurl{10.1007/s10990-008-9032-6}}


\bibitem[Skalka et~al\mbox{.}(2020)]%
        {skalka2020types}
\bibfield{author}{\bibinfo{person}{Christian Skalka}, \bibinfo{person}{David
  Darais}, \bibinfo{person}{Trent Jaeger}, {and} \bibinfo{person}{Frank
  Capobianco}.} \bibinfo{year}{2020}\natexlab{}.
\newblock \showarticletitle{Types and abstract interpretation for authorization
  hook advice}. In \bibinfo{booktitle}{\emph{2020 IEEE 33rd Computer Security
  Foundations Symposium (CSF)}}. IEEE, \bibinfo{pages}{139--152}.
\newblock
\href{https://doi.org/10.1109/CSF49147.2020.00018}{doi:\nolinkurl{10.1109/CSF49147.2020.00018}}


\bibitem[Skalka et~al\mbox{.}(2008)]%
        {skalka2008types}
\bibfield{author}{\bibinfo{person}{Christian Skalka}, \bibinfo{person}{Scott
  Smith}, {and} \bibinfo{person}{David Van~Horn}.}
  \bibinfo{year}{2008}\natexlab{}.
\newblock \showarticletitle{Types and trace effects of higher order programs}.
\newblock \bibinfo{journal}{\emph{Journal of Functional Programming}}
  \bibinfo{volume}{18}, \bibinfo{number}{2} (\bibinfo{year}{2008}).
\newblock
\href{https://doi.org/10.1017/S0956796807006466}{doi:\nolinkurl{10.1017/S0956796807006466}}


\bibitem[Suenaga(2008)]%
        {suenaga2008type}
\bibfield{author}{\bibinfo{person}{Kohei Suenaga}.}
  \bibinfo{year}{2008}\natexlab{}.
\newblock \showarticletitle{Type-based deadlock-freedom verification for
  non-block-structured lock primitives and mutable references}. In
  \bibinfo{booktitle}{\emph{Asian Symposium on Programming Languages and
  Systems}}. Springer, \bibinfo{pages}{155--170}.
\newblock
\href{https://doi.org/10.1007/978-3-540-89330-1_12}{doi:\nolinkurl{10.1007/978-3-540-89330-1_12}}


\bibitem[Tang and Jouvelot(1994)]%
        {tang1994separate}
\bibfield{author}{\bibinfo{person}{Yan~Mei Tang} {and} \bibinfo{person}{Pierre
  Jouvelot}.} \bibinfo{year}{1994}\natexlab{}.
\newblock \showarticletitle{Separate abstract interpretation for control-flow
  analysis}. In \bibinfo{booktitle}{\emph{International Symposium on
  Theoretical Aspects of Computer Software}}. Springer,
  \bibinfo{pages}{224--243}.
\newblock
\href{https://doi.org/10.1007/3-540-57887-0_98}{doi:\nolinkurl{10.1007/3-540-57887-0_98}}


\bibitem[Tate(2013)]%
        {tate2013sequential}
\bibfield{author}{\bibinfo{person}{Ross Tate}.}
  \bibinfo{year}{2013}\natexlab{}.
\newblock \showarticletitle{The Sequential Semantics of Producer Effect
  Systems}. In \bibinfo{booktitle}{\emph{{POPL} '13: Proceedings of the 40th
  annual {ACM} { SIGPLAN-SIGACT} symposium on Principles of Programming
  Languages}}. \bibinfo{publisher}{{ACM}}.
\newblock
\href{https://doi.org/10.1145/2429069.2429074}{doi:\nolinkurl{10.1145/2429069.2429074}}


\bibitem[Timany et~al\mbox{.}(2024)]%
        {timany2024logical}
\bibfield{author}{\bibinfo{person}{Amin Timany}, \bibinfo{person}{Robbert
  Krebbers}, \bibinfo{person}{Derek Dreyer}, {and} \bibinfo{person}{Lars
  Birkedal}.} \bibinfo{year}{2024}\natexlab{}.
\newblock \showarticletitle{A logical approach to type soundness}.
\newblock \bibinfo{journal}{\emph{J. ACM}} \bibinfo{volume}{71},
  \bibinfo{number}{6} (\bibinfo{year}{2024}), \bibinfo{pages}{1--75}.
\newblock
\href{https://doi.org/10.1145/3676954}{doi:\nolinkurl{10.1145/3676954}}


\bibitem[van Dooren and Steegmans(2005)]%
        {vanDooren2005}
\bibfield{author}{\bibinfo{person}{Marko van Dooren} {and}
  \bibinfo{person}{Eric Steegmans}.} \bibinfo{year}{2005}\natexlab{}.
\newblock \showarticletitle{Combining the Robustness of Checked Exceptions with
  the Flexibility of Unchecked Exceptions Using Anchored Exception
  Declarations}. In \bibinfo{booktitle}{\emph{Proceedings of the ACM SIGPLAN
  Conference on Object-oriented Programming, Systems, Languages, and
  Applications}}. \bibinfo{numpages}{17}~pages.
\newblock
\showISBNx{1-59593-031-0}
\href{https://doi.org/10.1145/1094811.1094847}{doi:\nolinkurl{10.1145/1094811.1094847}}


\bibitem[Vouillon and Jouvelot(1995)]%
        {vouillon1995type}
\bibfield{author}{\bibinfo{person}{J\'er\^ome Vouillon} {and}
  \bibinfo{person}{Pierre Jouvelot}.} \bibinfo{year}{1995}\natexlab{}.
\newblock \bibinfo{title}{Type and Effect Systems via Abstract Interpretation}.
\newblock
\urldef\tempurl%
\url{http://cri.ensmp.fr/classement/doc/A-273.pdf}
\showURL{%
\tempurl}


\bibitem[Yetter(1990)]%
        {yetter1990quantales}
\bibfield{author}{\bibinfo{person}{David~N Yetter}.}
  \bibinfo{year}{1990}\natexlab{}.
\newblock \showarticletitle{Quantales and (noncommutative) linear logic}.
\newblock \bibinfo{journal}{\emph{The Journal of Symbolic Logic}}
  \bibinfo{volume}{55}, \bibinfo{number}{01} (\bibinfo{year}{1990}),
  \bibinfo{pages}{41--64}.
\newblock
\href{https://doi.org/10.2307/2274953}{doi:\nolinkurl{10.2307/2274953}}


\end{thebibliography}

\end{document}